\documentclass[journal,draftcls,onecolumn,12pt,twoside]{IEEEtran}

\usepackage{amsmath,amsfonts}
\usepackage{amsthm}
\usepackage{algorithmic}
\usepackage{algorithm}
\usepackage{array}
\usepackage[caption=false,font=normalsize,labelfont=sf,textfont=sf]{subfig}
\usepackage{textcomp}
\usepackage{stfloats}
\usepackage{url}
\usepackage{verbatim}
\usepackage{graphicx}
\usepackage{cite}
\usepackage{bm}

\newtheorem{theorem}{Theorem}
\newtheorem{lemma}{Lemma}

\newtheorem{define}{Definition}

\newtheorem{remark}{Remark}

\newcommand{\norm}[1]{\left\|#1\right\|}
\newcommand{\norms}[1]{\norm{#1}^2}
\newcommand{\br}[1]{\left(#1\right)}
\newcommand{\brs}[1]{\left[#1\right]}
\newcommand{\brc}[1]{\left\{#1\right\}}
\newcommand{\card}[1]{\left|#1\right|}

\newcommand{\Gens}{\mathcal{G}\br{M,n,P}}
\newcommand{\Bens}{\mathcal{B}\br{M,n,P}}

\newcommand\prob[1]{\Pr\left[#1\right]}

\newcommand\rvec[1]{\mathsf{\mathbf{#1}}}

\newcommand\dvec[1]{\bm{#1}}
\newcommand\dset[1]{\mathcal{#1}}

\newcommand\dmat[1]{\bm{#1}}

\newcommand{\bz}{\mathbf{z}}
\newcommand{\bcm}{\mathbf{c}_{\dset{M}}}
\newcommand{\bcf}{\mathbf{c}_{\dset{F}}}
\newcommand{\bcc}{\mathbf{c}_{\dset{C}}}
\newcommand{\bzero}{\mathbf{0}}
\newcommand{\bI}{\mathbf{I}}
\newcommand{\be}{\pmb{\eta}}
\newcommand{\bA}{\mathbf{A}}
\newcommand{\bB}{\mathbf{B}}

\newcommand{\sEt}{\dset{U}_{e}\br{t}}
\newcommand{\sEe}{\dset{E}_{e}}
\newcommand{\sEr}{\dset{E}_{r}}
\newcommand{\sEcr}{\dset{E}^{c}_{r}}
\newcommand{\sUe}{\dset{U}_{e}\br{t}}
\newcommand{\sUer}{\dset{U}_{e,r}\br{t}}
\newcommand{\sUcr}{\dset{U}^{c}_{r}\br{t}}

\begin{document}

\title{Energy Efficiency of Unsourced Random Access over the Binary-Input Gaussian Channel}

\author{Anton Glebov, Pavel Rybin, Kirill Andreev, Alexey Frolov
\thanks{Anton Glebov, Pavel Rybin, Kirill Andreev and Alexey Frolov are with the Center for Next Generation Wireless and IoT (NGW), Skolkovo Institute of Science and Technology, Moscow, Russia (emails: anton.glebov@skoltech.ru, p.rybin@skoltech.ru, k.andreev@skoltech.ru, al.frolov@skoltech.ru).}
\thanks{The research was carried at Skolkovo Institute of Science and Technology and supported by the Russian Science Foundation (project no. 18-19-00673), \protect\url{https://rscf.ru/en/project/18-19-00673/}}
}

\maketitle

\begin{abstract}
We investigate the fundamental limits of the unsourced random access over the binary-input Gaussian channel. By fundamental limits, we mean the minimal energy per bit required to achieve the target per-user probability of error. The original method proposed by Y.~Polyanskiy (2017) and based on Gallager's trick does not work well for binary signaling. We utilize Fano's method, which is based on the choice of the so-called ``good'' region. We apply this method for the cases of Gaussian and binary codebooks and obtain two achievability bounds. The first bound is very close to Polyanskiy's bound but does not lead to any improvement. At the same time, the numerical results show that the bound for the binary case practically coincides with the bound for the Gaussian codebook. Thus, we conclude that binary modulation does not lead to performance degradation, and energy-efficient schemes with binary modulation do exist.
\end{abstract}


\section{Introduction}

The paper deals with the \textit{unsourced random access} (URA) -- a new paradigm proposed in \cite{polyanskiy2017perspective} for massive machine type communications (mMTC). The challenges posed by mMTC are (a) the need to handle a large number of devices, (b) short data packets, and (c) sporadic access. URA employs grant-free transmission, meaning that packets are sent to the network without prior communication with the base station, and it uses the same encoder for all users (or equivalently, the same codebook). This decision is supported by the fact that using different encoders would increase the receiver complexity as the decoder would have to first identify which encoders were used. The same codebook assumption decouples the identification and decoding tasks, and the URA receiver only has to decode the messages without identifying the users who sent them.

In~\cite{polyanskiy2017perspective}, the URA problem is stated from an information-theoretic perspective. The fundamental limits for the Gaussian multiple-access channel (MAC) are given in~\cite{polyanskiy2017perspective}. The asymptotic improvement is presented in \cite{ZPT-isit19}. Both papers consider the case where the number of active users is fixed and known at the receiver. In~\cite{achRandUser2023}, these results are extended to the scenario where the number of active users is random and unknown at the receiver. Various low-complexity URA schemes have been proposed for the Gaussian MAC. The main approaches are $T$-fold irregular repetition slotted ALOHA protocol with collision resolution \cite{ordentlich2017low, vem2017user, Marshakov2019Polar}, sparse interleave division multiple access (IDMA) scheme \cite{AmalladinneJoint2023}, random spreading and correlation-based energy detection \cite{AmalladinnePolar2020, Duman2021, PradhanLDPCSC2021} and coded compressed sensing \cite{codedCS2020, fengler2020sparcs, AmalladinneAMP2021, Amalladinne2022}. The scheme \cite{PradhanLDPCSC2021} offers state-of-the-art performance for the Gaussian MAC. More realistic channel models such as single-antenna quasi-static Raleigh fading channel and MIMO channel were also considered in the literature. In this paper, we focus on the Gaussian MAC and refer the reader to \cite{FASURA2022, Andreev2022} for more details on the fading channels. 

We note that a large number of practical schemes (e.g. \cite{ordentlich2017low, vem2017user, Marshakov2019Polar, AmalladinneJoint2023}) utilize binary signaling (or BPSK modulation). At the same time, such schemes show worse energy efficiency in comparison to the schemes \cite{AmalladinnePolar2020, Duman2021, PradhanLDPCSC2021} using Gaussian codebooks. Thus, the natural question arises: Does binary signaling restrict energy efficiency? To the best of our knowledge, this question was not considered in the literature. In what follows, we aim to establish fundamental limits for binary-input Gaussian MAC. The original method from \cite{polyanskiy2017perspective} based on Gallager's $\rho$ trick does not work well for the binary-input Gaussian MAC. The reason is the need to calculate several consecutive expectations over false codewords, missed codewords, and noise, which seem to be a challenging problem for the binary signaling case. We were inspired by the approach proposed in \cite{Poor2022ML} to analyze the energy efficiency of the URA with the MIMO receiver. This method involves introducing a ``good'' region for estimating the per-user probability of error (PUPE), which allows us to calculate the probability of erroneous events inside this region and use the probability of getting outside this region as an estimate of the remaining probability of erroneous events. Our contribution is as follows. We start with a Gaussian codebook and obtain a new energy efficiency achievability bound for the URA over the Gaussian MAC. This bound is close to the bound from \cite{polyanskiy2017perspective} (the gap is no more than $0.06$ dB) but does not lead to any improvements. At the same time, the new method allows us to analyze binary signaling and derive achievability bound for the binary codebook. Numerical experiments show that two obtained bounds practically coincide. Therefore, we conclude that the URA with a binary codebook achieves almost the same energy efficiency as with a Gaussian codebook for the given set of system parameters, i.e. binary signaling does not impose any additional limitations.  

The paper is structured as follows. Section~\ref{sec:syst_model} describes the system model. Section~\ref{sec:bounds} presents the approach based on Fano's method, conditions for decoding error and region, and achievability bounds for the Gaussian URA channel with Gaussian and binary codebooks. In Section~\ref{sec:num_res} we show the numerical results for the bounds obtained in this paper and the Polyanskiy bound and analyze the results.

\textsl{Notation}: in this paper, we refer to the set of real numbers as $\mathbb{R}$. Scalar values are represented by non-boldface letters such as $x$ or $X$, while vectors are represented by boldface small letters such as $\dvec{x}$ (here we consider vectors as columns), and matrices are represented by boldface capital letters such as $\dmat{X}$. The identity matrix of size $n \times n$ is denoted by $\dmat{I}_n$. The Euclidean norm of a vector $\dvec{x}$ is denoted by $\norm{\dvec{x}}$. Calligraphic letters such as $\dset{X}$ are used to denote sets and ensembles. We reserve the letters $\mathcal{E}$ and $\mathcal{U}$ for events and use $\mathcal{E}^c$ to denote the complementary event to $\mathcal{E}$. Greek letters are used for real numbers, auxiliary coefficients, and functions, such as $\alpha$ and $\phi(\cdot)$. For any positive integer $n$, we use the notation $[n] = \{1,\dots,n\}$. We use the notation $\mathcal{N}\br{\bm{0},\dmat{I}_n}$ to represent a standard normal random vector. Lastly, we use the symbol $\mathbb{E}$ to denote the expectation operator, and $\prob{\mathcal{E}}$ to represent the probability of the event $\mathcal{E}$.

\section{System model}
\label{sec:syst_model}

This paper focuses on the partial activity model proposed in~\cite{polyanskiy2017perspective}. In this model, there are a total of $K_{tot}$ users, but only $K_{a}$ of them are active at any given time, where $K_{a} \ll K_{tot}$. We assume that the receiver knows the number of active users $K_{a}$ and that communication is done in a frame-synchronized fashion, where the length of each frame is $n$ channel uses. Each active user transmits $k$ bits within a frame, where $M = 2^{k}$. All users use the same message set $[M]$ and encoder function $f: [M] \to \mathbb{R}^n$, subject to the power constraint $\norm{f(W)}^2 \leq nP$ for all $W \in [M]$. Let $\dset{S} = \{W_1, \ldots, W_{K_a}\}$ be the set of transmitted messages. The channel output $\rvec{y} \in \mathbb{R}^n$ is given by
\begin{equation}
\label{eq:original_cs}
\rvec{y}=\sum_{i=1}^{K_a} f(W_i) +\rvec{z},
\end{equation}
where $\rvec{z} \sim \mathcal{N}\br{\bzero,\bI_{n}}$ is an additive white Gaussian noise (AWGN).

The goal of decoding is to output a set $\hat{\dset{S}} \subseteq [M]$ of size $K_{a}$, representing the set of transmitted messages up to permutation. The main performance measure used is the Per User Probability of Error (PUPE)
$$
P_e = \frac{1}{K_a} \sum\limits_{i=1}^{K_a} \prob{W_i \not\in \hat{\dset{S}}}.
$$

The focus of this paper is to minimize the energy-per-bit ($E_b/N_0 = Pn/k$, where $k=\log_2M$) spent by each user, subject to the restriction that $P_{e} \leq \epsilon$, for a given message set $[M]$, the length of the coded sequence $n$, and the number of active users $K_{a}$. The paper considers two ensembles of random-access codes that satisfy these conditions:
\begin{define}
Let $\Gens$ be the ensemble of Gaussian codebooks of size $n \times M$, where each element is sampled i.i.d. from $\mathcal{N}\br{0,P}$.
\end{define}
\begin{define}
Let $\Bens$ be the ensemble of binary codebooks of size $n \times M$, where each element is sampled i.i.d. from $\brc{-\sqrt{P}, \sqrt{P}}$ with probability $1/2$.
\end{define}

\section{Achievability bounds}
\label{sec:bounds}

In this section, we obtain the achievability bounds for the Gaussian URA channel with Gaussian and binary codebooks. First, we introduce some common notations and lemmas used in the corresponding proofs. Recall that $\dset{S}$ and $\hat{\dset{S}}$ are the sets of transmitted and received messages. Let $\dset{M} = \dset{S} \backslash \hat{\dset{S}}$, $\dset{F} = \hat{\dset{S}} \backslash \dset{S}$ and $\dset{C} = \dset{S} \bigcap \hat{\dset{S}}$ be the sets of missed, falsely-detected and correctly received messages, respectively. For $\dset{I} \subseteq [M]$ we introduce
$$
\mathbf{c}_{\dset{I}} = \sum\limits_{W \in \dset{I}}f\br{W}.
$$
In what follows, we assume that $t$ denotes the number of missed and false codewords. We note that $\card{\dset{M}} = \card{\dset{F}} = t$ since the receiver knows the number of active users $K_a$ and $\card{\hat{\dset{S}}} = K_{a}$.

Let $\sEt = \{|\dset{M}| = t \}$ denote the $t$-error event. Following \cite{polyanskiy2017perspective} we estimate $P_e$ in the following way
$$
P_e \leq \sum\limits_{t=1}^{K_{a}}\frac{t}{K_{a}}\prob{\sEt} + p_{0},
$$
where
\[
p_0 = \prob{\left[ \bigcup\limits_{i \ne j} \{ W_i = W_j \} \right] \cup \left[ \bigcup\limits_{i \in [K_a]} \{ \norm{f(W_i)}^2 > Pn \} \right]}
\]
is the probability to select the same codewords from the codebook by different users or to have the power of the signal more than $Pn$. Clearly, the second condition is needed for the Gaussian codebook only.

Let us consider the error event in more detail. In the case of maximum likelihood (ML) decoding, we get the decoding error if there exists the set of messages $\hat{\dset{S}}$ such that the sum of the corresponding codewords is closer to the received sequence $\rvec{y}$ than the sum of the codewords corresponding to the transmitted messages $\dset{S}$:
$$
\norms{\rvec{y} - \mathbf{c}_{\dset{S}}} \geq \norms{\rvec{y} - \mathbf{c}_{\hat{\dset{S}}}}.
$$
As a result, we get the following decoding error condition $\norms{\bz} \geq \norms{\bcm - \bcf + \bz}$. So, let us introduce the following error event:
$$
\sEe = \brc{\norms{\bz} - \norms{\bcm - \bcf + \bz} \geq 0}.
$$
Estimating the error probability based only on this condition we have to apply union bound on the last step. One knows it leads to the overestimation of the probability. A good method to reduce this overestimation is to use Fano's method. For this purpose, we introduce the ``good'' region (some restrictions on the noise norm) where we apply the above error condition:
$$
\norms{\bz} < \alpha\norms{\bcm + \bz} + \beta n.
$$
\begin{remark}
Let us provide some comments on the choice of ``good'' region:
\begin{itemize}
\item adding of $\bcc$ to the region description does not lead to performance improvement,
\item the region description does not contain $\bcf$ to avoid a huge binomial coefficient $\binom{M-K_a}{t}$.
\end{itemize}
\end{remark}
Thus, let us introduce the ``good'' region event:
$$
\sEr = \brc{\alpha\norms{\bcm + \bz} - \norms{\bz} + \beta n > 0}
$$
and the ``bad'' (complementary) one:
$$
\sEcr = \brc{\norms{\bz} - \alpha\norms{\bcm + \bz} - \beta n \geq 0}.
$$
Thus, we can introduce the following unions:
$$
\sUe = \bigcup\limits_{\dset{M}, \dset{F}: \card{\dset{M}} = \card{\dset{F}} = t}\sEe, 
$$
$$
\sUer = \bigcup\limits_{\dset{M}, \dset{F}: \card{\dset{M}} = \card{\dset{F}} = t}\sEe \cap \sEr,
$$
$$
\sUcr = \bigcup\limits_{\dset{M}: \card{\dset{M}} = t}\sEcr.
$$
In \cite{polyanskiy2017perspective}, the main result was obtained by estimating the $\prob{\sUe}$ using union bound with a combination of several of Gallager's $\rho$-tricks. In this paper, we estimate the error probability in the following way:
$$
\prob{\sUe} \leq \inf\limits_{\alpha > 0, \beta > 0} \br{ \prob{\sUer} + \prob{\sUcr} },
$$
where we apply union bound to estimate $\prob{\sUer}$ and $\prob{\sUcr}$
$$
\prob{\sUer} \leq \binom{M-K_a}{t}\binom{K_{a}}{t}\prob{\sEe \cap \sEr},
$$
$$
\prob{\sUcr} \leq \binom{K_{a}}{t}\prob{\sEcr}.
$$
To find good estimates for $\prob{\sEe \cap \sEr}$ and $\prob{\sEcr}$ we use the Chernoff bound \cite{cover2012elements}  and the Theorem 3.2a.1 from \cite{mathai1992quadratic}, which we formulate as follows
\begin{lemma}[Chernoff bound, \cite{cover2012elements}]\label{lemma:chernoff}
Let $\chi_1$ and $\chi_2$ be any random variables, then for any $u,v > 0$ the following bounds hold:
\[
\prob{\chi_1 \geq 0} \leq \mathbb{E}_{\chi_1} \left[ \exp\left( u \chi_1 \right) \right]
\]
and
\[
\prob{\chi_1 \geq 0, \chi_2 \geq 0} \leq \mathbb{E}_{\chi_1, \chi_2} \left[ \exp\left( u \chi_1 + v \chi_2 \right) \right].
\]
\end{lemma}

\begin{lemma}[Theorem 3.2a.1, \cite{mathai1992quadratic}]\label{lem:quad_form_rv}
Let $\be \sim \mathcal{N}(\bzero, \bI_n)$, $\mathbf{b}$ be an arbitrary vector of length $n$, $\bA$ be a symmetric matrix of size $n \times n$ and $\lambda_{min}(\bI_{n} - 2\bA) > 0$ then
\begin{flalign*}
&\mathbb{E}_{\be}\exp\brs{ \be^{T}\bA\be + \mathbf{b}^{T}\be} = \nonumber \\
&= \exp\brs{-\frac{1}{2}\log\det\br{\bI_{n} - 2\bA} + \frac{1}{2}\mathbf{b}^{T}\br{\bI_{n}-2\bA}^{-1}\mathbf{b}}. \nonumber
\end{flalign*}
\end{lemma}

\subsection{Gaussian codebook}
\begin{theorem}
\label{th:gauss_codebook}
Fix $P' < P$. There exists a random-access code with Gaussian codebook from $\mathcal{G}\br{M,n,P'}$ for $K_a$-user Gaussian URA channel satisfying power-constraint $P$ and having
$$
P_e \leq \sum\limits_{t=1}^{K_{a}}\frac{t}{K_{a}}\inf\limits_{\alpha, \beta > 0}\left[ \binom{M - K_{a}}{t}\binom{K_{a}}{t}p_{t} + \binom{K_{a}}{t}q_{t} \right] + p_{0},
$$
where
$$
p_{0} = \frac{\binom{K_{a}}{2}}{M} + K_{a}\prob{\frac{1}{n}\norms{\bz} > \frac{P^{'}}{P}},
$$
$$
p_{t} = \inf\limits_{u,v > 0}\exp\brs{-\frac{n}{2}\log\det\br{\bI_{3} - 2\bA} + v\beta n},
$$
$$
q_{t} = \inf\limits_{\delta > 0}\exp\brs{-\frac{n}{2}\log\det\br{\bI_{2} - 2\delta\bB} - \delta\beta n},
$$
where
$$
\bA = \br{\begin{array}{rrr}
\br{\alpha - 1}v & \br{\alpha v - u}\sqrt{P't} & u\sqrt{P't} \\
\br{\alpha v - u}\sqrt{P't} & \br{\alpha v - u} P't & uP't \\
u\sqrt{P't} & uP't & -uP't \\
\end{array}},
$$
$$
\bB = \br{\begin{array}{rr}
1 - \alpha & -\alpha\sqrt{P't} \\
-\alpha\sqrt{P't} & -\alpha P't \\
\end{array}},
$$
and the following conditions hold
$$
\lambda_{min}\br{\bI_{3} - 2bA} > 0\text{ and }\lambda_{min}\br{\bI_{2} - 2\delta\bB} > 0.
$$
\end{theorem}

\begin{proof}
Consider the ensemble $\mathcal{G}\br{M,n,P'}$. Let us introduce the following vector
$$\be^{T} = \br{\bz^{T},\, \bcm^{T}/\sqrt{P't},\, \bcf^{T}/\sqrt{P't}}
$$
$\be \sim \mathcal{N}(\bzero, \bI_{3n})$. Using this notation we can rewrite the events $\sEe$, $\sEr$ and $\sEcr$ in the following way:
$$
\sEe = \brc{\be^{T}\bA_{e}\be \geq 0},
$$
$$
\sEr = \brc{\be^{T}\bA_{r}\be + \beta n > 0},
$$
$$
\sEcr = \brc{\be^{T}\bA^{c}_{r}\be - \beta n \geq 0},
$$
where $\bA^{c}_{r} = -\bA_{r}$ and
$$
\bA_{e} = \br{\begin{array}{rrr}
     \bzero & -\sqrt{P't}\bI_{n} & \sqrt{P't}\bI_{n}  \\
     -\sqrt{P't}\bI_{n} & -P't\bI_{n} & P't\bI_{n} \\
     \sqrt{P't}\bI_{n} & P't\bI_{n} & -P't\bI_{n}
\end{array}},
$$
$$
\bA_{r} = \br{\begin{array}{rrr}
     \br{\alpha - 1}\bI_{n} & \alpha\sqrt{P't}\bI_{n} & \bzero  \\
     \alpha\sqrt{P't}\bI_{n} & \alpha P't\bI_{n} & \bzero \\
     \bzero & \bzero & \bzero 
\end{array}}.
$$

Now, we can obtain the estimation for the probability $\prob{\sEe \cap \sEr}$ using Chernoff bound for the joint events:
$$
\prob{\sEe \cap \sEr} \leq \inf\limits_{u,v > 0}\mathbb{E}_{\be}\exp\brs{\be^{T} \bA_{n} \be + v\beta n },
$$
where $\bA_{n} = u\bA_{e} + v\bA_{r}$, and using the Lemma~\ref{lem:quad_form_rv} (with $\mathbf{b}$ equal to zero vector) we obtain the expectation $\mathbb{E}_{\be}$:
$$
\prob{\sEe \cap \sEr} \leq \inf\limits_{u,v > 0}\exp\brs{-\frac{1}{2}\log\det\br{\bI_{3n} - 2\bA_{n}} + v\beta n}
$$
for $\lambda_{min}\br{\bI_{3n} - 2\bA_{n}} > 0$.

Let us note, that the matrix $\bA_{n}$ can be represented as $\bI_{n} \otimes \bA$ by reordering its columns and rows. So, we can represent $\br{\bI_{3n} - 2\bA_{n}}$ as $\bI_{n} \otimes \br{\bI_{3} - 2\bA}$ and it is clear that $-\frac{1}{2}\log\det\br{\bI_{3n} - 2\bA_{n}} = -\frac{n}{2}\log\det\br{\bI_{3} - 2\bA}$. Thus, $\prob{\sEe \cap \sEr} \leq p_{t}$.

Similarly, we have $\prob{\sEcr} \leq q_{t}$ for $\lambda_{min}\br{\bI_{2} - 2\delta \bB} > 0$. It is true due to the Laplace (cofactor) expansion over the last block-column (block-row) of matrix $\br{\bI_{3n} - 2\delta \bA^{c}_{r}}$ and row and column reordering, that we use above.

Finally, applying union bound we get the estimate for $P_e$ with $p_0$ defined earlier.
\end{proof}

\begin{remark}
We note that one can consider regions including $\bcf$ with the use of the ideas from \cite{ZPT-isit19}, i.e. let us consider the following region
\[
\sEr = \{ \norm{\mathbf{c}_{\dset{F}}} \geq \sqrt{\beta n}  \},
\]
we have
\begin{eqnarray*}
\Pr\left[ \sUcr \right] &=& \Pr\left[ \bigcup_{\dset{F} \subset [M]\backslash[K_a]}\{ \norm{\mathbf{c}_{\dset{F}}} < \sqrt{\beta n}  \} \right] \\ &=& \Pr\left[ \min_{\dset{F} \subset [M]\backslash[K_a]} \norm{\mathbf{c}_{\dset{F}}} < \sqrt{\beta n} \right].
\end{eqnarray*}
The latter term can be upper bounded with the use of concentration properties and Gordon's inequality \cite{gordon1988} for the expectation of the minimum of the Gaussian process. Thus, the techniques from \cite{ZPT-isit19} can be reformulated as an instance of Fano's method. At the same time, we note that we were not able to get any improvement for the finite-length regime.
\end{remark}

\subsection{Binary codebook}
\begin{theorem}
\label{th:binary_codebook}
There exists a random-access code with binary codebook from $\Bens$ for $K_a$-user Gaussian URA channel satisfying power-constraint $P$ and having
$$
P_e \leq \sum\limits_{t=1}^{K_{a}}\frac{t}{K_{a}}\inf\limits_{\alpha, \beta > 0}\left[ \binom{M-K_{a}}{t}\binom{K_{a}}{t}p_{t} + \binom{K_{a}}{t}q_{t}  \right] + p_0,
$$
where $p_0 = \frac{\binom{K_a}{2}}{M}$ and
\begin{flalign*}
&p_t = \inf\limits_{u,v > 0}\left[ \exp\brs{-n\zeta\br{\alpha, \beta, v}} \times \begin{array}{cc}
&  \\
&  \\
&
\end{array} \right. \nonumber \\
&\times \left. \br{\sum\limits_{m=-t}^{t}\sum\limits_{f=-t}^{t} \rho_{m}\rho_{f}\exp{\brs{P\phi\br{\alpha, u, v, m, f}}}}^{n} \right] \nonumber
\end{flalign*}
where
$$
\zeta\br{\cdot} = \frac{1}{2}\log\br{1 - 2v\br{\alpha - 1}} - \beta v,
$$
$$
\phi\br{\cdot} = \br{2\frac{\br{\alpha v m - u\br{m-f}}^2}{1 - 2v\br{\alpha - 1}} + \alpha v m^{2} - u\br{m-f}^2},
$$
and
\begin{flalign*}
& q_{t} = \inf\limits_{\delta > 0}\left[ \exp\brs{-n\xi\br{\alpha, \beta, \delta}} \times \begin{array}{cc}
&  \\
&  \\
&
\end{array} \right. \nonumber \\
&\times \left. \br{\sum\limits_{m = -t}^{t} \rho_{m}\exp\brs{P\psi\br{\alpha, \delta, m}}}^{n} \right], \nonumber
\end{flalign*}
where
$$
\xi\br{\cdot} = \frac{1}{2}\log\br{1 - 2\delta\br{1 - \alpha}} + \delta\beta,
$$
$$
\psi\br{\cdot} = \delta\alpha\frac{2\delta - 1}{1 - 2\delta\br{1 - \alpha}}m^{2},
$$
and where
$$
\rho_{i} =
\begin{cases}
\binom{t}{\frac{1}{2}\br{i+t}}2^{-t}, & i \in \brc{2j-t,\,\forall j: 0 \leq j \leq t} \\
0, & \textit{otherwise}
\end{cases}
$$
and the following conditions hold
$$
1 - 2v\br{\alpha - 1} > 0 \text{ and } 1-2\delta\br{1 - \alpha} > 0.
$$
\end{theorem}

\begin{proof}
Opening the brackets of the norms, let us rewrite the events $\sEe$, $\sEr$ and $\sEcr$ in the following way:
$$
\sEe = \brc{-2\br{\bcm - \bcf}^{T}\bz - \norms{\bcm - \bcf} > 0},
$$
$$
\sEr = \brc{\br{\alpha - 1}\norms{\bz} + 2\alpha \bcm^T\bz + \alpha\norms{\bcm} + \beta n > 0},
$$
$$
\sEcr = \brc{\br{1 - \alpha}\norms{\bz} - 2\alpha \bcm^T\bz - \alpha\norms{\bcm} - \beta n > 0}.
$$
Thus, we can estimate the probability $\prob{\sEe \cap \sEr}$ using Chernoff bound for the joint events in the following way:
\begin{flalign*}
&\prob{\sEe \cap \sEr} \leq \inf\limits_{u,v > 0}\mathbb{E}_{\bz,\bcm,\bcf}\exp\left[\br{\alpha - 1}v\norms{\bz} \right. \nonumber \\ 
&\left. + 2\br{\alpha v \bcm^T - u\br{\bcm -\bcf}^T}\bz \right. \nonumber \\ 
&\left. + \alpha v \norms{\bcm} - u\norms{\bcm - \bcf} + \beta v n \right].
\end{flalign*}
First, let us obtain the expectation value over $\bz$ considering $\bcm$ and $\bcf$ fixed. For this purpose we apply the Lemma~\ref{lem:quad_form_rv}:
\begin{flalign*}
&\prob{\sEe \cap \sEr} \leq \inf\limits_{u,v > 0}\mathbb{E}_{\bcm, \bcf}\exp\left[-\frac{n}{2}\log\br{1 - 2\br{\alpha - 1}v} \right. \nonumber \\ 
&\left. + 2\frac{\norms{\alpha v \bcm - u\br{\bcm - \bcf}}}{1 - 2\br{\alpha - 1}v} \right. \nonumber \\
&\left. + \alpha v \norms{\bcm} - u\norms{\bcm - \bcf} + \beta v n \right]
\end{flalign*}
To obtain the expectation value over $\bcm$ and $\bcf$, let us consider the $\bcm = \br{c_{\dset{M},1}, c_{\dset{M},2}, \dots, c_{\dset{M},n}}$ (for $\bcf$ the following arguments are similar), that is the sum of $t$ codewords (with the power constraint $P$ for each coordinate). Then, $c_{\dset{M},l} \in \brc{\br{2j-t}\sqrt{P}, \forall j: 0 \leq j \leq t}$, $\forall l \in \brs{n}$, because changing the sign of $j$ terms change the sum by $2j\sqrt{P}$ in absolute value, e.g. if all of $t$ terms have the value $-\sqrt{P}$ and we select random $j$ terms to change their value to $\sqrt{P}$, then the sum will be $\br{2j-t}\sqrt{P}$. So, the number of ways to obtain the sum $c_{\dset{M},l} = \br{2j - t}\sqrt{P}$, $\forall j: 0 \leq j \leq t$, is $\binom{t}{j}$ and overall number of possible combinations of $t$ terms with two possible values for each is equal to $2^{t}$. Thus, we can introduce the probability of $c_{\dset{M},l}$ equals to $i\sqrt{P}$ as
$$
\rho_{i} = \prob{c_{\dset{M},l} = i\sqrt{P}} = \frac{1}{2^{t}}\binom{t}{\frac{i+t}{2}}
$$
for $i \in \brc{2j-t, \forall j: 0 \leq j \leq t}$ and $0$ otherwise. Now, let us consider the term 
$$
\exp\brs{\norms{\bcm}} = \exp\brs{\sum\limits_{l=1}^{n}c^{2}_{\dset{M},l}} = \prod\limits_{l = 1}^n\exp\brs{c^{2}_{\dset{M},l}}.
$$
Let $c_{\dset{M},l} = i_{l}\sqrt{P}$, then the probability of such vector $\bcm$ is equal to $\prob{\norms{\bcm} = \sum\limits_{l=1}^{n}i^{2}_{l}P} = \prod\limits_{l=1}^{n}\rho_{i_{l}}$. Thus,
\begin{flalign*}
&\mathbb{E}\brs{\exp\brs{\norms{\bcm}}} = \sum\limits_{\norms{\bcm}}\prob{\norms{\bcm}}\exp\brs{\norms{\bcm}} \nonumber \\
& = \sum\limits_{\br{i_{1}, i_{2}, \dots, i_{n}}}\prod\limits_{l=1}^{n}\rho_{i_{l}}\exp\brs{i^{2}_{l}P} = \br{\sum\limits_{i=-t}^{t}\rho_{i}\exp\brs{i^{2}P}}^{n}.
\end{flalign*}
So, applying this estimate over $\bcm$ and $\bcf$ we get $\prob{\sEe \cap \sEr} \leq p_{t}$. Similarly, we have $\prob{\sEcr} \leq q_{t}$.
\end{proof}

\section{Numerical results}
\label{sec:num_res}
This section presents numerical results for the achievable bounds derived in the previous section for the Gaussian URA channel with Gaussian and binary codebooks. We assume a communication system with a frame length of $n = 30000$ channel usages, where each user transmits $k = 100$ bits. Fig.\ref{fig:ebno_ka} shows the energy efficiency plotted against the number of active users $K_a$ for different schemes. Energy efficiency is defined as the minimum $E_b/N_0$ such that the probability of error per user is $P_e < 0.05$. As can be seen in Fig.\ref{fig:ebno_ka}, the results obtained for Gaussian (Theorem\ref{th:gauss_codebook}) and binary (Theorem~\ref{th:binary_codebook}) codebooks are almost identical and very close to the original achievability bound from \cite{polyanskiy2017perspective} (e. g. for $250$ active users the values are $1.210$dB, $1.211$dB and $1.154$dB, respectively). Now let us consider the simulation results for the energy efficiency of some practical schemes from \cite{AmalladinneJoint2023, AmalladinnePolar2020, PradhanLDPCSC2021}. We must note that the schemes from \cite{AmalladinnePolar2020, PradhanLDPCSC2021} use the Gaussian codebooks even though they are based on binary polar and LDPC codes, respectively, because they use spreading sequences based on Gaussian signaling for transmission. In contrast, the scheme of~\cite{AmalladinneJoint2023} uses binary codebooks and obtains results for a binary input AWGN channel with BPSK modulation. As you can see, the Gaussian codebooks for practical schemes provide much better energy efficiency than binary ones. However, our achievability bound shows that the binary codebook does not impose any additional restrictions on the energy efficiency for the considered set of parameters. As a further research direction, it will be interesting to apply Fano's method to the case of a random number of active users and compare the results with \cite{achRandUser2023}.

\begin{figure}
\centering
\includegraphics{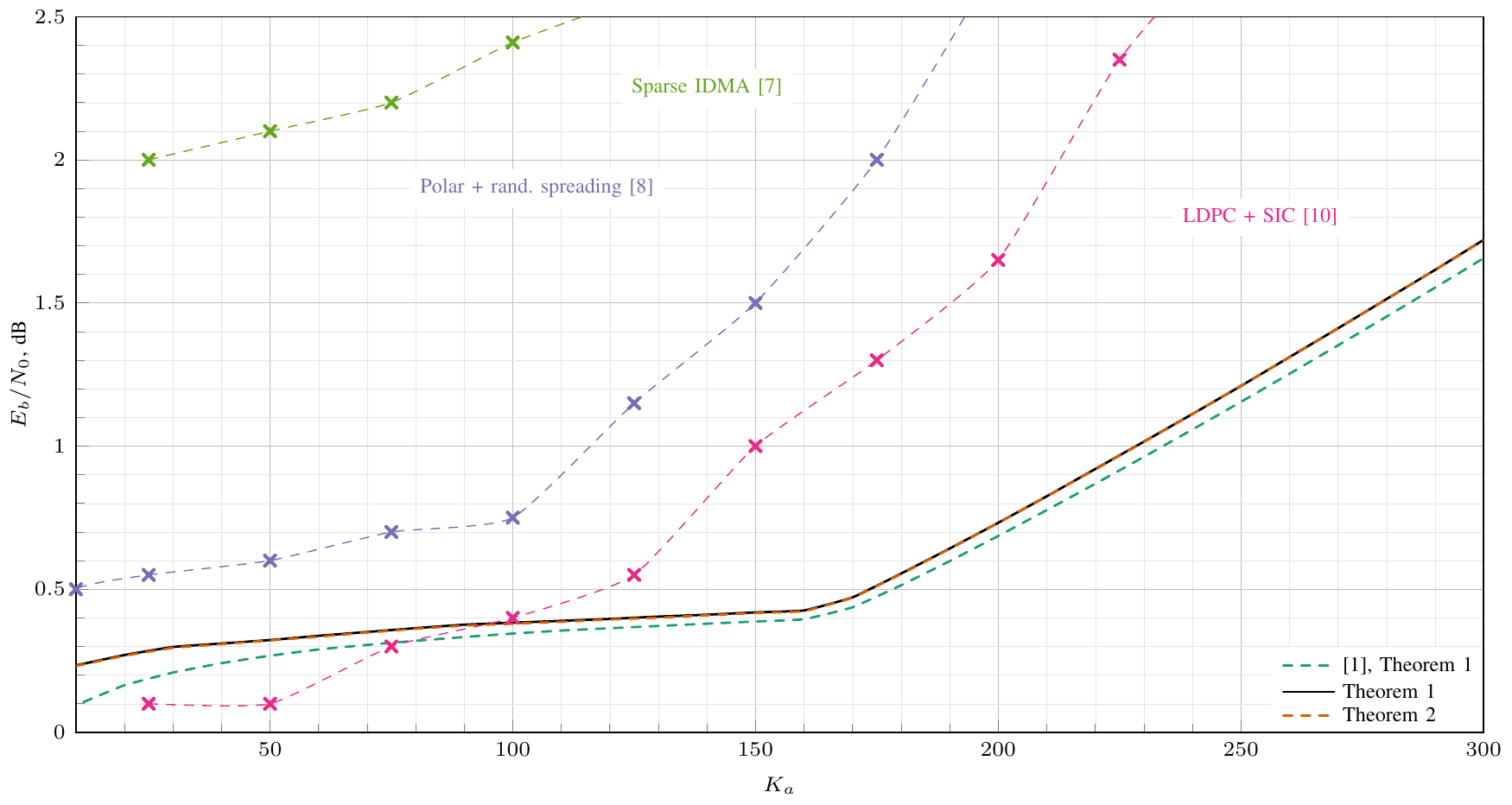}
\caption{AWGN MAC achievability bounds for $k=100$ information bits, frame length $n=30000$, $P_e=0.05$.}
\label{fig:ebno_ka}
\end{figure}

\bibliographystyle{IEEEtran}
\bibliography{main}

\vfill

\end{document}